\def\submission{0}
\ifnum\submission=0
\documentclass[11pt]{article}
\else
\documentclass[9pt]{article}
\usepackage{multicol}
\usepackage{float}
\fi

\usepackage[margin=1.0in]{geometry}

\usepackage{amsthm, color}
\usepackage{amsmath}
\usepackage{comment}
\usepackage{mathtools}
\usepackage{hyperref}
\usepackage{thmtools,thm-restate}

\DeclarePairedDelimiter\abs{\lvert}{\rvert}

\theoremstyle{plain}
\newtheorem{thm}{Theorem}[section]
\newtheorem{lma}[thm]{Lemma}
\theoremstyle{definition}
\newtheorem{defn}[thm]{Definition} 
\newtheorem{cla}[thm]{Claim}

\newtheorem{cor}[thm]{Corollary}

\newtheorem{remark}[thm]{Remark} 
\newcommand{\Mod}[1]{\ \mathrm{mod}\ #1}
\newcommand{\polylog}{{\mathrm{polylog}}}

\title{On the Hardness of Massively Parallel Computation}
\ifnum\submission=0
\author{Kai-Min Chung \\ Academia Sinica \\ kmchung@iis.sinica.edu.tw
    \and Kuan-Yi Ho \\ University of Texas at Austin \\ kyho@cs.utexas.edu
    \and Xiaorui Sun \\ University of Illinois at Chicago \\ xiaorui@uic.edu}
\else
\author{}
\fi
\date{}

\begin{document}
\maketitle

\begin{abstract}
We investigate whether there are inherent limits of parallelization in the (randomized) massively parallel computation (MPC) model by comparing it with the (sequential) RAM model. As our main result, we show the existence of hard functions that are essentially not parallelizable in the MPC model. Based on the widely-used random oracle methodology in cryptography with a cryptographic hash function $h:\{0,1\}^n \rightarrow \{0,1\}^n$ computable in time $t_h$, we show that there exists a function that can be computed in time $O(T\cdot t_h)$ and space $S$ by a RAM algorithm, but any MPC algorithm with local memory size $s < S/c$ for some $c>1$ requires at least $\tilde{\Omega}(T)$\footnote{Throughout the paper, we use the convention $\tilde{\Omega}(T)=\Omega(T/\polylog (T))$ and $\tilde{O}(T)=O(T\cdot\polylog (T))$} rounds to compute the function, even in the average case, for a wide range of parameters $n \leq S \leq T \leq 2^{n^{1/4}}$. Our result is almost optimal in the sense that by taking $T$ to be much larger than $t_h$, \textit{e.g.}, $T$ to be sub-exponential in $t_h$, to compute the function, the round complexity of any MPC algorithm with small local memory size is asymptotically the same (up to a polylogarithmic factor) as the time complexity of the RAM algorithm.
Our result is obtained by adapting the so-called compression argument from the data structure lower bounds and cryptography literature to the context of massively parallel computation.

\end{abstract}

\thispagestyle{empty}
\newpage
\pagenumbering{arabic}
\def\Line{\mathbf{Line}}
\def\SimLine{\mathbf{SimLine}}
\def\Dec{\mathbf{Dec}}
\def\Enc{\mathbf{Enc}}

\ifnum\submission=1
\begin{multicols}{2}
\fi

\section{Introduction}
In the last decade, there has been significant development of parallel computation. Modern parallel computation frameworks such as MapReduce, Hadoop, and Spark are designed to manipulate large-scale data sets and share a number of common properties. Modeling and analyzing these frameworks algorithmically help us confirm the practical success theoretically and the new ideas created in the process may also be adapted to enhance the performance of these practical frameworks. To this end, much effort has been made to model the essential properties behind these frameworks and explore the power of the developed model. 

The first theoretic model capturing the modern parallel computation frameworks has been proposed by Karloff, Suri, and Vassilvitskii~\cite{KSV10}. Ever since then, the theoretical study of these frameworks started to increase rapidly and several refinements of the model along with new algorithms have been proposed~\cite{ANOY14, CLMMOS18,IMS17, KSV10,KMVV15, LMSV11,MZ15,RVW16}. 

The Massively Parallel Computation (MPC) model consists of $m$ machines and each of them has local memory of size $s$. The input is partitioned arbitrarily across all the machines and the computation proceeds in synchronized rounds. In each round, each machine is able to do any computation on its own memory. After the computation is done, each machine computes a set of messages. Then, the underlying system will direct the messages to the corresponding machine. As, in practice, most costs come from the network communication, the goal is to minimize the round complexity. Also, to rule out the trivial algorithm, common constraints require that given input of size $N$, $ms=\Theta(N)$ and $N^\epsilon\leq m\leq N^{1-\epsilon}$ for some constant $\epsilon>0$.


Many computational problems, such as 
 graph problems~\cite{AG18,ANOY14,
AndoniSSWZ18,
andoni2019log,
Assadi17,
ABBMS17,AK17, arXiv,BKV12, 
BehnezhadDTK18,
behnezhad2019exponentially,
BrandtFU18,
ChangFGUZ19,
CLMMOS18,
GhaffariGMR18,
GhaffariU19,
gamlath2018weighted,
LackiMW18,LMSV11,
Onak18,RastogiMCS13},
clustering~\cite{BMVKV12,BBLM14,EIM11, ghaffari2019improved,YV17} and 
submodular function optimization~\cite{PENW16,EN15, KMVV15, MKSK13},
have been studied in this model, with an emphasis on developing algorithms that minimize the number of communication rounds. 
For example, recently,~\cite{IMS17} showed that massively parallel computation can simulate dynamic programming algorithms admitting two properties, i.e.  monotonicity and decomposability. This shows the power of massively parallel computation since the process of dynamic programming typically requires large memory space and  is considered inherently sequential. 
\paragraph{Limitation of Massively Parallel Computation.} In this work, we investigate whether there are inherent limits in the massively parallel computation model. Namely, whether there are functions that are hard to parallelize in the MPC model. Towards this, we compare it with the (sequential) RAM model of computation. Suppose we have a function that can be computed by a RAM algorithm with time complexity $T$ and space complexity $S$ (assume the input size $N\leq S$). It is easy to see that an MPC algorithm can compute the function in $T$ rounds by emulating the RAM computation step by step, even when each machine has $O(\log S)$ local memory size. Also,  if each machine has local memory size $S$, then trivially the function can be computed in one round. Therefore, to show such a limitation, ideally, we would like to show the existence of a function computable in time $T$ and space $S$ in the RAM model but it requires $\Omega(T)$ rounds to compute for any MPC algorithms with local memory size $s < S$. We refer to this as the \emph{best-possible hardness} for the MPC model.

The hardness of the MPC model has been investigated by the seminal work of Roughgarden, Vassilvitskii, and Wang~\cite{RVW16}, who showed that there are functions requiring $\Omega(\log_s N)$ rounds to compute in the MPC model\footnote{In fact, the lower bound holds in a stronger model called $s$-shuffle circuits}. This gives a logarithmic lower bound when the local memory size $s = O(1)$, but only a constant lower bound for the typical settings where $s$ is polynomial in $N$. Nevertheless, \cite{RVW16} showed that an $\omega(\log_s N)$ round complexity lower bound in the MPC model for any problems in \textbf{P} implies $\textbf{P} \neq  \textbf{NC}^1$, which is beyond the reach of the current techniques in complexity theory. Thus, the $\Omega(\log_s N)$ lower bound is essentially the best we can hope for given the status of complexity theory if we look for unconditional lower bounds.

To circumvent the barrier in complexity theory, we borrow ideas from cryptography. Specifically, we investigate the hardness of the MPC model in the Random Oracle (RO) model based on the widely used random oracle methodology in cryptography, which we briefly review as follows.

\paragraph{Random Oracle Methodology.} 
 This is a popular methodology for designing cryptographic constructions, which consists of the following two steps. First, we consider the Random Oracle (RO) model where all parties have oracle access to a truly random function $\mathsf{RO}:\{0,1\}^n \rightarrow \{0,1\}^n$. We design and prove the security for a cryptographic construction in the (idealized) RO model. Next, we replace the random oracle by a  ``good cryptographic hashing function'' $h$ (such as SHA3) to obtain a concrete construction, and \emph{assume} that the construction has the same security as the ``ideal'' one analyzed in the RO  model. This methodology is widely used in both practice and theory to obtain more efficient and simpler constructions~\cite{BR93,BR94,Fischlin05,MMV11} or to achieve stronger security and new feasibility results~\cite{BCTW16,BHR12,Pass03}. Of course, replacing the oracle by a hash function is merely a heuristic. The validity of such heuristic has been investigated in the literature, where several counterexamples (i.e., constructions that are proven secure in the RO model but become insecure when the RO is instantiated by any concrete hash functions) are known~\cite{BBP04,CGH04-1,CGH04,GT03,MRH04,N02}. However, these counterexamples are contrived in the sense that they are constructed for this purpose, instead of obtaining a useful cryptographic construction. For all natural constructions, the heuristic holds so far and sometimes the proved security in the RO model matches the best-known attacks~\cite{CDGS18,DTT10,DGK17}. Indeed, the random oracle methodology is well-accepted in practice where many RO-based constructions (e.g., RSA-OAEP) have been used for years as part of the standard in practical cryptographic systems~\cite{ACPRT17,BR94}.

\subsection{Our Results and Techniques}\label{result}

We demonstrate a limitation of the MPC model by establishing a nearly best-possible hardness result in the Random Oracle model. 
Let $\mathsf{RO}:\{0,1\}^{n}\to\{0,1\}^{n}$ be a random oracle where making a query to $\mathsf{RO}$ takes $O(n)$ time.
Specifically, in the following theorem, we show the existence of a function in the RO model that can be computed in time $O(T\cdot n)$ and space $S$ by a RAM algorithm such that any MPC algorithm with local memory size $s \leq O(S)$ requires at least $\tilde{\Omega}(T)$ rounds to compute the function, even in the average case, for a wide range of  parameters $T$ and $S$. 

\begin{thm}[Nearly Best-Possible Hardness in the RO Model]\label{line_inf}
 There exists a universal constant $c>1$ such that for any sufficiently large $n>0$, the following holds.  For any $n\leq S<2^{O(n^{1/4})}$, $S\leq T<2^{O(n^{1/4})}$, there is an oracle function $f^\mathsf{RO}:\{0,1\}^S\to\{0,1\}^n$ such that it can be computed using memory of size $O(S)$ in $O(T\cdot n)$ time by a RAM computation with access to $\mathsf{RO}$, but for any (potentially randomized) massively parallel computation algorithm $\mathcal{A}^\mathsf{RO}$ with $m< 2^{O(n^{1/4})}$ machines, local memory of size $s$ where $s\leq S/c$, and the number of local queries per round $q<2^{n/4}$ to $\mathsf{RO}$, the probability that $\mathcal{A}^\mathsf{RO}$ computes $f^{\mathsf{RO}}$ correctly in $o(T/\log^2 T)$ rounds is at most $1/3$ over the random choice of $\mathsf{RO}$ and input.
\end{thm}

In particular, for any parameters $T$ and $S$, by setting $n = \polylog(T)$, Theorem~\ref{line_inf} shows the existence of an oracle function computable in time $\tilde{O}(T)$ and space $O(S)$ by a RAM algorithm where any MPC algorithm with local memory size sufficiently smaller than $S$ requires $\tilde{\Omega}(T)$ rounds to compute the function. 
As discussed above, this is the best-possible hardness up to a poly-logarithmic factor. 
Furthermore, following the random oracle methodology, we can instantiate the random oracle with a good cryptographic hash function $h$ with time complexity $t_h = \mathrm{poly}(n)$. We then obtain a concrete hard function $f^h:\{0,1\}^S\to\{0,1\}^n$ that can be  computed in time $\tilde{O}(T)$ and space $O(S)$ by a RAM algorithm, yet assuming the validity of the random oracle methodology, $f^h$ is hard to compute for any (randomized) MPC algorithm with local memory of size $s\leq S/c$ for some constant $c>1$.\footnote{For this conclusion we would need to set $n = \polylog(T)$ and requires $h$ to have sub-exponential hardness. While such assumption is strong, it is commonly assumed in practical RO-based cryptographic systems.} This is the best-possible hardness up to a poly-logarithmic factor. Note that the hardness holds even when the total memory size $ms \gg S$ as long as the local memory size is bounded. 

We remark that the way our hard function $f^\mathsf{RO}$ makes use of the random oracle is quite standard and analogous to several existing cryptographic constructions (e.g.,~\cite{ACKKPT16,ACPRT17,MMV11}), so it is unlikely that the random oracle methodology fails to apply to our hard function $f^\mathsf{RO}$ (see more discussion in Section~\ref{related}). 
Thus, one way to interpret our result is that either $f^h$ indeed shows a fundamental limitation of parallelization in the MPC model, or gives a natural counter-example for the random oracle methodology (which would be surprising).

In the following, we describe the hard function we consider and explain the intuition of its hardness. To help illustrate the idea, let us start with a warm-up (hard) function, which we analyze formally in Appendix~\ref{warmup} for the sake of completeness. Let $\mathsf{RO}:\{0,1\}^{n}\to\{0,1\}^{n}$ be a random oracle and let $T,u,v$ be the parameters specified later. Consider the function $\SimLine^\mathsf{RO}_{n,T,u,v}\footnote{The name $\SimLine$ stands for ``simple-line,'' which is in contrast to the hard function $\Line$  we introduce below.}:\{0,1\}^{uv}\to\{0,1\}^n$ defined as follows. The input is parsed as $v$ strings $x_i\in\{0,1\}^u$ for all $i\in [v]$. On input $x=x_1,x_2,...,x_v$, the output of $\SimLine^\mathsf{RO}_{n,T,u,v}(x)$ is defined by iteratively applying $\mathsf{RO}$ as follows. Let $r_1=0^u$, and 
\begin{align*}
    &(r_{i+1},z_{i+1}) \coloneqq \mathsf{RO}(x_{i\Mod v},r_i,0^*),\quad \forall i\in[T],
\end{align*}
the output of $\SimLine_{n,T, u, v}^\mathsf{RO}(x)$ is defined as the answer to the last query, $(r_{T+1},z_{T+1})$. In other words, we can view $\SimLine_{n,T,u,v}$ as defined by a line of $T$ nodes, where the first node is associated with initial values $r_1=0^u$, and for each $i\in [w]$, the values of next node $i+1$ is obtained by querying the oracle on $(x_{i\Mod v},r_i,0^*)$. 

To get some intuition of the hardness, first note that since the oracle is random, intuitively, the only way to learn the output $(r_{T+1},z_{T+1})$ is to make queries to learn the value of each node (i.e., $(r_i,z_i)$) in order, which requires to know the corresponding input $x_{i\Mod v}$. However, since the local memory of each machine is bounded by $s$, intuitively, a machine can only store at most  $s/u$ inputs $x_i$'s. Thus, in each round, the machines can only learn the value of at most $s/u$ new nodes. Therefore, a MPC algorithm with local space $s$ would need  $\Omega(Tu/s)$ rounds to compute $\SimLine_{n,T, u, v}^\mathsf{RO}(x)$.

%

Formalizing the above intuition, however, is non-trivial, since the algorithm may encode the inputs arbitrarily and also the machines may collaborate in an arbitrary way.  Thus, it would be difficult to formalize what information is stored in the machine's memory and how much the algorithm learns about the line directly. For the warm-up case of $\SimLine_{n,T, u, v}^\mathsf{RO}$, we can formalize this intuition by a rather standard use of the so-called ``compression argument'', a powerful technique for establishing lower bounds in both data structure and cryptography literature~\cite{ACPRT17,AS15,DGK17,DTT10,LN18,PD06}.
To give some intuition about the argument, consider an MPC algorithm computing the function one by one along the line. For each round, the queries of a machine must contain the corresponding input $x_{i\Mod v}$ in order to proceed along the line. Thus, by examining the queries of this machine, we can infer which part of the input is stored in the local memory. Now, if a machine learns too many nodes in one round, it reveals that it stores many $x_i$'s in its small local memory. The key of the compression argument is to show that this would allow us to compress the input $x$ and the random oracle $\mathsf{RO}$ beyond the information-theoretic limit, which is a contradiction. We defer the formal analysis of $\SimLine_{n,T, u, v}^\mathsf{RO}$ to Appendix~\ref{warmup}.

Note that $\SimLine_{n,T, u, v}^\mathsf{RO}$ can be computed in time $O(Tn)$ by a RAM program, but above we only argue a $\Omega(Tu/s)$, instead of $\tilde{\Omega}(T)$, lower bound for the round complexity of MPC algorithms. To prove the desired hardness result, we make the function harder by letting each node take a random input $x_{\ell_i}$ instead of $x_{i\Mod v}$, where the index $\ell_i$ is specified by the random oracle (like $r_i$). 

More precisely, we now describe our hard function $\Line^\mathsf{RO}_{n,T,u,v}:\{0,1\}^{uv}\to\{0,1\}^n$ as follows (see Figure~\ref{linegraph} for a pictorial illustration). The input is parsed as $v$ strings $x_i \in \{0,1\}^u$ for $i\in [v]$. On input $x=x_1,x_2,...,x_v$, the output of $\Line^\mathsf{RO}_{n,T,u,v}(x)$ is defined by iteratively applying $\mathsf{RO}$ as follows. Let $\ell_1 = 1$ and $r_1=0^u$, and 
\begin{align*}
    &(\ell_{i+1},r_{i+1},z_{i+1}) \coloneqq \mathsf{RO}(i,x_{\ell_i},r_i,0^*),\quad \forall i\in[T],
\end{align*}
the output of $\Line_{n,w, u, v}^\mathsf{RO}(x)$ is defined as the answer to the last correct query $(\ell_{T+1},r_{T+1},z_{T+1})$. Similarly, we can view $\Line_{n,T, u, v}$ as defined by a line of $T$ nodes, where the first node is associated with initial values $\ell_1 = 1$ and $r_1=0^u$, and for each $i\in [T]$, the values of next node $i+1$ is obtained by using $\ell_i$ to select an input $x_{\ell_i}$ and query the oracle on $(i,x_{\ell_i},r_i,0^*)$. Clearly, the function can be evaluated with $O(uv)$ space and $O(Tn)$ time by following the evaluation over the line. Thus, for given parameters $S$ and $T$, we can set $v = S/u$  to get a function with RAM complexity specified in Theorem~\ref{line_inf}.

Now, intuitively, since $s \leq S/c$ for a constant $c$,  a machine can only store a constant fraction of $x_i$'s, and since $\ell_i$'s are random, the probability that a machine can learn the value of $k$ new nodes should decay exponentially in $k$. Thus, ``with high probability,'' a MPC algorithm can learn at most, say, $\log^2 T$ new nodes, and hence it would require $\tilde{\Omega}(T)$ rounds to compute $\Line^\mathsf{RO}_{n,T,u,v}$, which gives us the desired hardness.

However, as above, formalizing this intuition is tricky, since a-priori there is no independence between the information stored by a machine and the indices $\ell_i$, and thus it is not clear whether the probability of learning $k$ new nodes decays exponentially in $k$. Roughly, we capture this intuition of exponential probability decay indirectly by a novel tweak to the compression argument. 
Specifically, in our proof, we enumerate all the oracles with different sequences of $k$ consecutive $\ell$'s, say $\ell_i,...,\ell_{i+k}$, and run the machine on these oracles. By considering all the queries obtained this way, we can formalize such exponential probability decay in the compression argument provided that $k$ is not too large, which allows us to show that a MPC algorithm can learn at most   $\log^2 T$ new nodes each round and hence require $\tilde{\Omega}(T)$ rounds to compute $\Line^\mathsf{RO}_{n,T,u,v}$.
As the argument is more involved, we defer a more detailed technical overview and the formal proof to Section~\ref{sec:main-thm}.

\subsection{Related Work}\label{related}
Massively Parallel Computation Model (a.k.a MapReduce Model) was proposed in~\cite{KSV10}, and 
has been refined and extended in~\cite{BKS13, ANOY14, RVW16}.
Many algorithmic techniques and problems have been studied in MPC model such as 
greedy algorithms~\cite{KMVV15}, 
dynamic programming~\cite{bateni2018massively,IMS17}, 
linear programming~\cite{assadi2019distributed},
graph algorithms~\cite{AG18,ANOY14,
AndoniSSWZ18,
andoni2019log,
Assadi17,
ABBMS17,AK17, arXiv,BKV12, 
BehnezhadDTK18,
behnezhad2019exponentially,
BrandtFU18,
ChangFGUZ19,
CLMMOS18,
GhaffariGMR18,
GhaffariU19,
gamlath2018weighted,
LackiMW18,LMSV11,
Onak18,RastogiMCS13},
clustering~\cite{BMVKV12,BBLM14,EIM11, ghaffari2019improved,YV17}, 
submodular function optimization~\cite{PENW16,EN15, KMVV15, MKSK13},
and query optimization~\cite{BKS13}.

The tradeoffs between space (total memory), communication and 
the number of communication rounds in MPC model have been studied. 
Pietracaprina et al.~\cite{PPRSU12} studied the space-round tradeoffs for certain kinds of matrix multiplication algorithm.
Afrati et al.~\cite{ASSU13} investigated the space-communication tradeoffs for single round algorithms.
Beame et al.~\cite{BKS13} studied the tradeoff between the amount of communication and the number of rounds.

Roughgarden et al.~\cite{RVW16} proved an $\lfloor \log_s n\rfloor$ round unconditional lower bounds. 
When local memory per machine $s$ is polynomially related to $n$, which is usually assumed in the MPC model, this gives a constant round lower bound. 
Fish et al.~\cite{FKLRT15} proved hierarchy theorems with respect to the computation time per processor.

Conditioned on the conjecture that graph connectivity cannot be solved in $o(\log n)$ communication rounds for memory per machine sublinear in the number of vertices,
Ghaffari et al.~\cite{GKU19} showed that
constant approximation of maximum matching,  vertex cover, and maximum independent set cannot be solved in $o(\log \log n)$ rounds, and the Lov{\'{a}}sz Local Lemma problem cannot be solved in $o(\log \log \log  n)$ rounds. 
Under the same conjecture, Yaroslavtsev et al.~\cite{YV17}
showed that single-linkage clustering cannot be approximated by constant factor in $o(\log n)$ rounds.
Recently, Nanongkai and Scquizzato showed that 
this conjecture is equivalent to the conjecture that log-space complete problem can not be solved in $o(\log n)$ rounds, and consequently, a large class of graph problems, such as single source shortest path, minimum cut, and planarity testing, require asymptotically the same number of
rounds under these assumptions~\cite{nanongkai2020equivalence}.


Another well-studied parallel computation model is the PRAM model. In this model, there is a polynomial number of processors and a shared memory. In each time (synchronized round), each processor can
read a constant number of memory cells from the shared memory, do some local computation, and write to a memory cell in the shared memory. Similar to the MPC model, it is known that super-logarithmic lower bound on the parallel time in the PRAM model implies strong circuit lower bounds. Therefore, some assumptions are needed in order to prove stronger lower bounds. 
To our knowledge, 
PRAM lower bound does not imply MPC lower bound in a trivial way, due to the local computation of MPC in every single round.
For example, Miltersen~\cite{Mil92} showed a strong lower bound in the random oracle model using a certain pointer jumping problem for PRAM. However, we note that the problem considered in~\cite{Mil92} is not hard in the MPC model. The reason is that in the MPC model, a local machine can make an arbitrary number of queries to the oracle in one round, and thus solve the problem considered in~\cite{Mil92} in one round.

The way that our hard function uses the random oracle is  analogous to several existing cryptographic constructions; in particular, the line of research in memory hard functions (MHFs)~\cite{ABP18,ACKKPT16,ACPRT17,AS15}. MHFs are hash functions whose evaluation cost is dominated by memory cost. MHFs found widespread applications such as password hashing, key derivation, and proofs-of-work, and some important candidates, e.g., $\mathsf{scrypt}$, are described in the RFC standard. The security (i.e., lower bounds on the so-called ``cumulative memory complexity'') of MHFs is analyzed in the RO model based on the random oracle methodology. As our construction uses RO in an analogous way as practically-used MHFs (both rely on sequential queries to the oracle), in our eyes, it would be quite surprising that our hard function becomes a counterexample to the random oracle methodology. 

Technically, our analysis is inspired by the analysis of MHFs, which also relies on the compression argument. However, we stress that the models are quite different and we cannot directly rely on the analysis of MHFs to establish the hardness of the MPC model. The main reason is that in the MPC model, the machines can make an \emph{arbitrary} number of adaptive queries to the oracle for free in one round, whereas the need of adaptive queries is the source of hardness for high cumulative memory complexity. Hence, our  $\Line^\mathsf{RO}_{n,T,u,v}$ function relies on a different reason (specifically, the fact that each machine is space bounded) to get hardness, and requires a different analysis from the MHFs.

The rest of the paper is organized as follow.  We define the massively parallel computation model in both the plain and the RO model in Section \ref{sec:model}. In Section \ref{sec:main-thm}, we state and prove our main theorem on the best-possible hardness for the MPC model.

\section{The Massively Parallel Computation Model} \label{sec:model}
Let $[n]=\{1,2,...,n\}$.
We adopt the Massively Parallel Computation Model (also known as MapReduce Model) according to \cite{KSV10}.
In this model, we are given a set of machines with a fixed size of local memory.
The input data is distributed across machines arbitrarily.
The computation proceeds in rounds. During a round, each machine runs a polynomial time algorithm on the data assigned to the machine. 
No communication between machines is allowed during a round.
Between rounds, machines are allowed to communicate so long as each machine receives no more communication than its memory.
Any data output from a machine must be computed locally from the data residing on the machine.

Formally, we define the (randomized) massively parallel computation with following parameters.
\ifnum\submission=0
\begin{table}[ht]
\begin{center}
\begin{tabular}{|p{8cm}|}
\hline
$s$: the local memory size for each machine\\
$m$: the number of machines\\
$N$: the size of the input\\
\hline
\end{tabular}
\end{center}
\vspace{-.6cm}
\caption{Parameters of massively parallel computation}
\label{tab}
\end{table}
\else
\begin{table}[H]
\begin{center}
\begin{tabular}{|p{7cm}|}
\hline
$s$: the local memory size for each machine\\
$m$: the number of machines\\
$N$: the size of the input\\
\hline
\end{tabular}
\end{center}
\caption{Parameters of massively parallel computation}
\label{tab}
\end{table}
\fi

\begin{defn}[Massively Parallel Computation]
A massively parallel computation consists of 
$m$ machines, local memory of size $s$ for each machine,
and a shared, read-only, and multiple access tape $\mathcal{T}$ containing an arbitrarily long random bit string.
The computation proceeds by round and starts from round $0$. Let $M_i^k \in \{0, 1\}^s$ be the local memory (input) of machine $i$ at the beginning of round $k$. Initially, the given input $x\in\{0,1\}^N$ is arbitrarily split and distributed among all the machines, i.e. each $M_i^0$ is assigned with an arbitrary partition of $x$.

In each round $k$, each machine $i$ runs a polynomial time algorithm $\mathcal{A}_{i}^k$ based on its local memory $M_i^k$ and the shared tape $\mathcal{T}$, and outputs $M_{i,j}^k$ to machine $j$ for all the $j \in [m]$. $\bigcup_{j\in [m]} M_{j,i}^k$ is the input of machine $i$ of $(k+1)$-th round.
\end{defn}
Note that in the above definition, it is required that the size of $\bigcup_{j\in [m]} M_{j,i}^k$ is smaller than $s$, the size of local memory. We refer to a MPC computation terminated at the end of round $R$ as a $R$-round MPC computation.


We consider the massively parallel computation with a random oracle.

\begin{defn}[Massively Parallel Computation with Oracle]
A massively parallel computation consists of 
$m$ machines, local memory of size $s$ for each machine,
 a shared, read-only, and multiple access tape $\mathcal{T}$ containing an arbitrarily long random bit string,
 and  a random oracle $\mathsf{RO}:\{0,1\}^h\to \{0,1\}^c$ which is uniformly drawn from all possible functions before the computation begins.

The computation proceeds by round and starts from round $0$. Let $M_i^k \in \{0, 1\}^s$ be the local memory (input) of machine $i$ at the beginning of round $k$. Initially, the given input $x\in\{0,1\}^N$ is arbitrarily split and distributed among all the machines, i.e. each $M_i^0$ is assigned with an arbitrary partition of $x$.

In each round $k$, each machine $i$ runs a polynomial time algorithm $\mathcal{A}_{i}^k$ based on its local memory $M_i^k$, the shared tape $\mathcal{T}$ and the (adaptive) queries of the random oracle, and outputs $M_{i,j}^k$ to machine $j$ for all the $j \in [m]$. $\bigcup_{j\in [m]} M_{j,i}^k$ is the input of machine $i$ of $(k+1)$-th round.
\end{defn}
\begin{remark}\label{ran}
    As a standard observation, in random oracle model, without loss of generality, we can \emph{consider only deterministic} MPC algorithm since the algorithm can use the randomness from the random oracle. In more detail,
    we can use a random oracle with a larger input domain and a deterministic MPC can simulate a randomized MPC by obtaining random bits from querying those extra oracle entries that are not used by the randomized MPC. Therefore, to establish a lower bound for the randomized MPC model, it suffices to consider deterministic MPC algorithms.
\end{remark}

\begin{defn}[Worst Case Correctness]
We say that a randomized $R$-round MPC computation (with random oracle) successfully computes a (oracle) function $f$ in worst case if for any input $x\in \{0,1\}^N$ which is arbitrarily distributed among the machines, the union of outputs of all the machines at the end of round $R$ is $f(x)$ with probability at least $\frac{1}{3}$ over the randomness of the MPC computation (and the random oracle). 
\end{defn} 
We also consider average case correctness.
\begin{defn}[Average Case Correctness]
We say that a randomized $R$-round MPC computation (with random oracle) successfully computes a (oracle) function $f$ in average case if given an input $x\in\{0,1\}^N$ which is drawn uniformly and arbitrarily distributed among the machines, the union of outputs of all the machines at the end of round $R$ is $f(x)$ with probability at least $\frac{1}{3}$ over the randomness of the input, the MPC computation (and the random oracle). 
\end{defn}

\section{Main Theorem} \label{sec:main-thm}
In this section, we state our main theorem. 
\begin{thm}\label{line}
There exists a universal constant $c>1$ such that for any sufficiently large $n>0$, let $\mathsf{RO}:\{0,1\}^n\to\{0,1\}^n$ be a random oracle and for any memory size $n\leq S<2^{O(n^{\frac{1}{4}})}$, and running time $S\leq T<2^{O(n^{\frac{1}{4}})}$, there is an oracle function $f^\mathsf{RO}:\{0,1\}^{S}\to\{0,1\}^n$ such that it can be computed in time $O(T\cdot n)$ using memory size $O(S)$ by a RAM algorithm in random oracle model. On the other hand, let $\mathcal{A}^O$ be a randomized massively parallel computation with $m<2^{O(n^{\frac{1}{4}})}$ machines, local memory of size $s\leq S/c$ and the number of local queries {$q<2^{n/4}$} to random oracle per round. Then, in random oracle model, $\mathcal{A}^\mathsf{RO}$ needs at least $ \tilde{\Omega}(T)$ rounds to compute $f^\mathsf{RO}$ even in average case. 
\end{thm}

The parameters are summarized in Table \ref{tab:tab2}.
\ifnum\submission=0

\begin{table}[ht]
\begin{center}
\begin{tabular}{|p{15cm}|}
\hline
$n$: the size of input and output of the random oracle\\
$S$: the memory size used by the RAM algorithm such that
$   n \leq S< 2^{O(n^{1/4})}$\\

$T$: the number of random oracle queries used  by RAM algorithm $S\leq T< 2^{O(n^{1/4})}$\\
$q$: the upper bound on the number of random oracle queries for every machine in each round
such that $q = 2^{O(n)}$\\
\hline
\end{tabular}
\end{center}
\caption{Parameters of Theorem~\ref{line}}
\label{tab:tab2}
\end{table}

\else

\begin{table}[H]
\begin{center}
\begin{tabular}{|p{7cm}|}
\hline
{
$n$: the size of input and output of the random oracle}\\
$S$: the memory size used by the RAM algorithm such that
$   n \leq S< 2^{O(n^{1/4})}$\\
$T$: the number of random oracle queries used  by RAM algorithm $S\leq T< 2^{O(n^{1/4})}$
\\
$q$: the upper bound on the number of random oracle queries for every machine in each round
such that $q = 2^{O(n)}$\\

\hline
\end{tabular}
\end{center}
\caption{Parameters of Theorem~\ref{line}}
\label{tab:tab2}
\end{table}

\fi



As we discussed in the introduction, for any parameters $T$ and $S$, by setting $n = \polylog(T)$ and instantiating the random oracle with a  cryptographic hash function $h$ with sub-exponential hardness, we obtain a concrete hard function $f^h:\{0,1\}^S\to\{0,1\}^n$ that can be  computed in time $\tilde{\Omega}(T)$ and space $O(S)$ by a RAM algorithm, yet assuming the validity of the random oracle methodology, $f^h$ is hard to compute for any (randomized) MPC algorithm with local memory of size $s\leq S/c$ for some constant $c>1$.

\ifnum\submission=0
\begin{table}[ht]
\begin{center}
\begin{tabular}{|p{15cm}|}
\hline
$u$: the size of each $x_i$ such that $u = n / 3$. $u$ is assumed to be large enough as otherwise, machine may guess it locally with non-trivial probability\\
$v$: the number of $x_i$'s in the input such that 
$v = S / u$\\
$w$: the number of iterations of the random oracle for the $\Line^{\mathsf{RO}}$ function such that 
$w = T$\\
$\ell_i$: $\lceil \log v \rceil$ bits of output of $(i-1)$-th iteration of the random oracle, which is used to specify the $x_{\ell_i}$ which is part of the input of $i$-th iteration of the random oracle\\
$r_i$: $u$ bits of the output of $(i-1)$-th iteration of the random oracle, which is used as part of the input of $i$-th iteration\\
$z_i$: redundant output of $(i-1)$-th iteration\\
\hline
\end{tabular}
\end{center}
\vspace{-.6cm}
\caption{Parameters of $\Line^{\mathsf{RO}}$ function given input $x_1, x_2, \dots, x_v$}
\label{tab:tab3}
\end{table}

\else

\begin{table}[H]
\begin{center}
\begin{tabular}{|p{7cm}|}
\hline
$u$: the size of each $x_i$ such that $u = n / 3$. $u$ is assumed to be large enough as otherwise, machine may guess it locally with non-trivial probability\\
$v$: the number of $x_i$'s in the input such that 
$v = S / u$\\
$w$: the number of iterations of the random oracle for the $\Line^{\mathsf{RO}}$ function such that 
$w = T$\\
$\ell_i$: $\lceil \log v \rceil$ bits of output of $(i-1)$-th iteration of the random oracle, which is used to specify the $x_{\ell_i}$ which is part of the input of $i$-th iteration of the random oracle\\
$r_i$: $u$ bits of the output of $(i-1)$-th iteration of the random oracle, which is used as part of the input of $i$-th iteration\\
$z_i$: redundant output of $(i-1)$-th iteration\\
\hline
\end{tabular}
\end{center}
\caption{Parameters of $\Line^{\mathsf{RO}}$ function given input $x_1, x_2, \dots, x_v$}
\label{tab:tab3}
\end{table}

\fi

Now we formally define the oracle function.
The parameters  are summarized in Table \ref{tab:tab3}.
Let $\Line_{n,w,u,v}^\mathsf{RO}:\{0,1\}^{uv}\to\{0,1\}^n$ be defined as follows: Given input $x=x_1,x_2,...,x_v$ such that $x_i\in\{0,1\}^u$ for all $i\in [v]$ and a random oracle $\mathsf{RO}:\{0,1\}^{n}\to\{0,1\}^{n}$, let $r_1=0^u$, $\ell_1=1$,
and 
\begin{align*}
    &(\ell_{i+1},r_{i+1},z_{i+1}) \coloneqq \mathsf{RO}(i,x_{\ell_i},r_i,0^*),\quad \forall i\in[w],
\end{align*}
the output of $\Line^\mathsf{RO}_{n,w, u, v}(x)$ is defined as the answer to the last correct query, $(\ell_{w+1},r_{w+1},z_{w+1})$. See Figure~\ref{linegraph} for an illustration.  Given the parameters $S,T$ in Theorem~\ref{line}, we set the parameters of $\Line^\mathsf{RO}$ as described in Table~\ref{tab}. 
\ifnum\submission=0
\begin{figure}
    \centering
    \includegraphics[scale = .90]{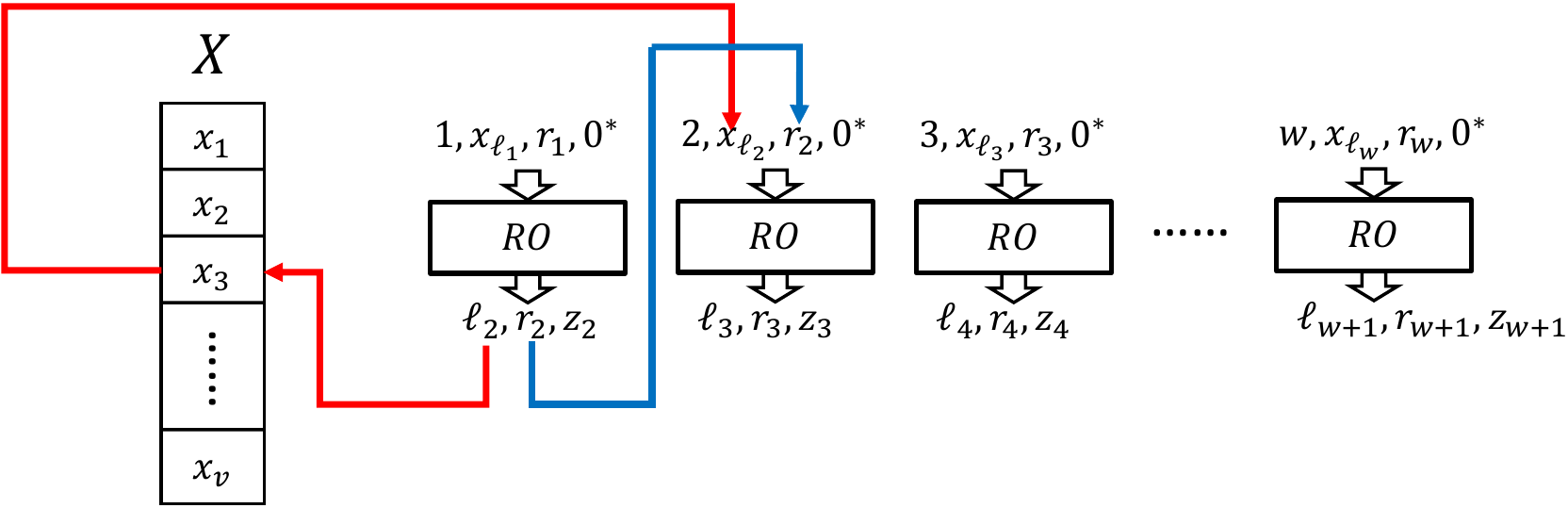}
    \caption{an illustration on how the $\Line^{\mathsf{RO}}$ is formed (suppose $\ell_2=3$). Each $RO$ box represents a correct random oracle query. 
    Note that each machine is not able to store the entire $X$.}
    \label{linegraph}
\end{figure}
\else
\begin{figure}[H]
    \centering
    \includegraphics[width=3.3in]{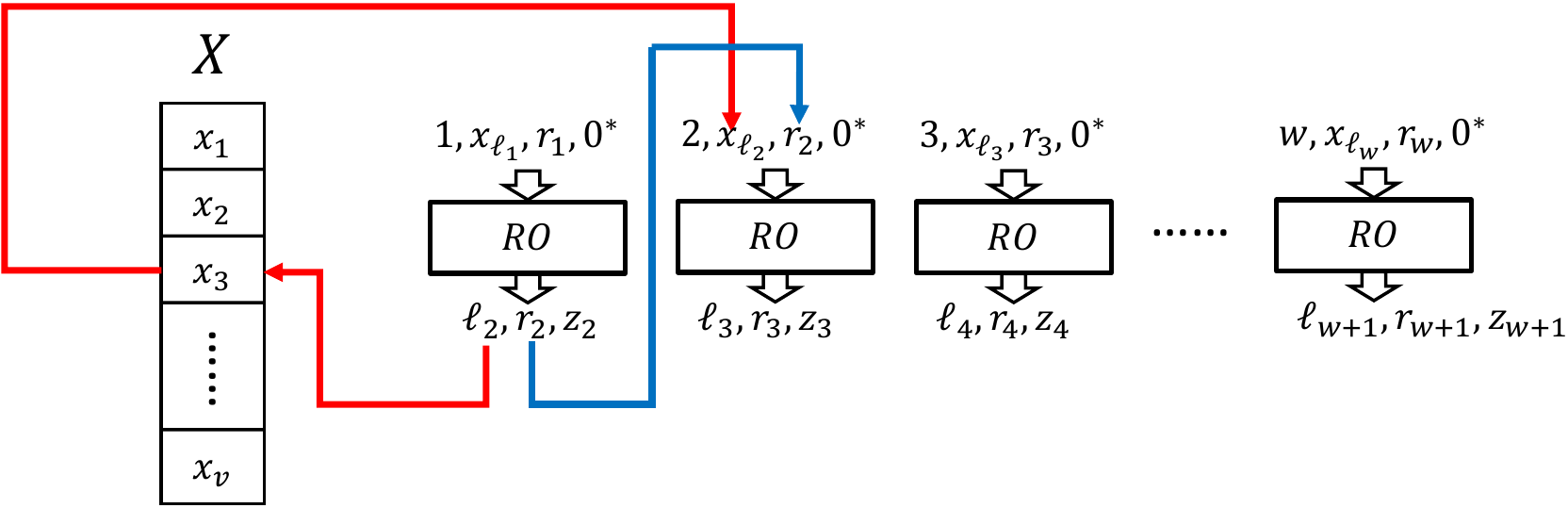}
    \caption{an illustration on how the $\Line^{\mathsf{RO}}$ is formed (suppose $\ell_2=3$). Each $RO$ box represents a correct random oracle query. 
    Note that each machine is not able to store the entire $X$.}
    \label{linegraph}
\end{figure}
\fi
 The obvious RAM algorithm already achieves the required performance on memory size and running time. As the standard observation in Remark~\ref{ran}, without loss of generality, we can assume the MPC computation is deterministic. Thus, Theorem~\ref{line} follows from the following lemma.

\begin{lma}\label{lmaline}
There exists a universal constant $c>1$ such that for any sufficiently large $n>0$, let $\mathsf{RO}:\{0,1\}^n\to\{0,1\}^n$ be a random oracle and for any $n\leq S<2^{O(n^{\frac{1}{4}})}$ and $S\leq T<2^{O(n^{\frac{1}{4}})}$, consider the function $\Line^\mathsf{RO}_{n,w,u,v}:\{0,1\}^{uv}\to\{0,1\}^n$ where $w=T$, $v=S/u$ and $u=n/3$. Let $\mathcal{A}^\mathsf{RO}$ be a deterministic massively parallel computation with $m<2^{O(n^{\frac{1}{4}})}$ machines, local memory of size $s\leq S/c$ and a number of 
at most $q<2^{n/4}$ random oracle queries per round per machine. Then, in random oracle model, $\mathcal{A}^\mathsf{RO}$ needs at least $R\geq \frac{w}{\log^2w}= \tilde{\Omega}(T)$ rounds to compute $\Line^\mathsf{RO}_{n,w,u,v}$ even in average case.
\end{lma}
As discussed in Section~\ref{result}, intuitively, the only way to learn $(\ell_{w+1},r_{w+1},z_{w+1})$ is to make queries to learn the value of each node (i.e.  $(\ell_{w+1},r_{w+1},z_{w+1})$) one by one. However, since $s\leq S/c$ for some constant $c$, intuitively, a machine can only store a constant fraction of $x_i$'s, and since $\ell_i$'s are random, the probability that a machine can learn the value of $p$ new nodes should decay exponentially in $p$.
 To formalize the above intuition, however, there are three issues. {First, we need to argue that indeed the algorithm can only query on the $\Line^\mathsf{RO}$ in the given order. Second, in general, an MPC algorithm is not restricted to store $x_i$'s in the most naive way; instead, it may encode them arbitrarily. Third, an MPC algorithm can query random oracle arbitrarily and thus the set of $x_i$'s stored would correlate with random oracle (in particular, $\ell_i$) arbitrarily. The first issue can be solved readily by a standard argument. To resolve the second issue, we employ the compression argument which helps us extract the information about $x_i$'s from the queries. Note that although the use of compression argument is inspired by the analysis of MHFs, we stress that the models are quite different and we cannot rely on the analysis of MHFs. The reason is that in MPC model, the machines can make an arbitrary number of adaptive queries in one round, whereas the need of adaptive queries is the source of hardness for MHFs.}


Now we give an overview of our technique. 
Observe that since computing $\Line^\mathsf{RO}$ requires the MPC algorithm to query on the $\Line^\mathsf{RO}$ in the given order, the queries of a machine must contain the corresponding $x_{\ell_{i-1}}$ to get $\ell_i$ and $r_i$ and this leaks which $x_i$'s a local machine stores. Thus, those $x_i$'s appearing in the queries can effectively represent those stored in the local memory. This motivates us to consider the set $B$ of $x_i$'s appearing in the queries. By applying compression argument, we can bound the size of $B$.
However, to get a better bound, we need to exploit the fact that each $\ell_j$ is uniformly random and independent. The set $B$ discussed above seems unlikely for us to do so since it depends on the random oracle (in particular, $\ell_j$). 
To remove the dependency on $\ell_j$'s,  we enumerate all the oracles with different sequences of $\log^2 w$ consecutive $\ell_j$'s in the encoding scheme and run the machine on these oracles. As the set of  queries obtained this way no longer depends on the enumerated $\ell_j$'s, it allows us to argue the probability of querying the next $p=\log^2 w$ correct queries decays exponentially in $p$.

We first formalize that the MPC algorithm can only query on the $\Line^\mathsf{RO}$ one by one. Given an oracle $\mathsf{RO}$, for each correct entry $j$ on $\Line^\mathsf{RO}$, we consider a set $V^{(j)}$ of oracle entries defined as follows. Initially, let $V^{(j)}=\phi$ and add $(j+1, x_{\ell_{j+1}}, r_{j+1})$ to $V^{(j)}$. Then, for $a_0=\ell_{j+1}$, any $a_1,...,a_{\log^2 w}\in [v]^{\log^2 w}$ and $b$ from $1$ to $\log^2 w$, add $(j+b+1, x_{a_b},r_b')$ to $V^{(j)}$ where $(\ell_b',r_b',z_b')\coloneqq \mathsf{RO}(j+b,x_{a_{b-1}},r_{b-1}')$, and we say $(j+b,x_{a_{b-1}},r_{b-1}')$ is the previous entry of $(j+b+1, x_{a_b},r_b')$. Note that for any $j$, $\abs{V^{(j)}}\leq v^{\log^2w}$. 
\begin{lma}\label{nojump2}
For any deterministic massively parallel computation $\mathcal{A}$ with $m$ machines and the number of queries $q$ running until the end of round $k$, let $E^{(k)}$ be the event that there exists a query position $t\in [(k+1)mq]$ and an oracle entry $e\in \bigcup_j V^{(j)}$ such that given $\mathcal{A}$ haven't queried the previous entry $e'$ of $e$, $\mathcal{A}$ successfully queries $e$ before the end of round $k$. Then, we have
\[
    \Pr_{(\mathsf{RO},X)}\left[E^{(k)}\right]\leq wv^{\log^2w}(k+1)mq2^{-u},
\]
where $\mathsf{RO}$ and $X$ are uniformly distributed.
\end{lma}
\begin{proof}
Fix some $t$ and $e\in \bigcup_jV^{(j)}$. Suppose all the previous queries are $q_1,...,q_{t-1}$ and the next query is $q_t$. We first fix $q_1,...,q_{t-1}$ and all the previous correct entries of $e'$ according to the $a_1,...,a_{\log^2 w}$ that let us add $e$ to $V^{(j)}$.
Then, we consider the set of random oracles, $\mathsf{RO}'$, consistent with the answers to the queries $q_1,q_2,...,q_{t-1}$ and the pre-fixed correct entries. For these oracles, the oracle entry $e'$ is well-defined and the oracle answer to $e'$ is still uniform over the set $\mathsf{RO}'$. Since the answers to the queries are fixed and $\mathcal{A}$ only depends on the answers to its queries, the next query $q_t$ will be the same for any oracle in $\mathsf{RO}'$. Moreover, $r_z'$ is still uniform over all $2^u$ possible values. Hence, we conclude that the guessing probability will be less than $2^{-u}$.
Thus, by a union bound, we obtain the claimed lemma.
\end{proof}

 For each $k$, let $E^{(k)}$ be the event defined in Lemma~\ref{nojump2} and $\overline{E^{(k)}}$ be the negation of $E^{(k)}$. Given a random oracle $\mathsf{RO}:\{0,1\}^n\to\{0,1\}^n$, an input $X\in\{0,1\}^{uv}$, and a massively parallel computation $\mathcal{A}$, let $j_k$ be the largest index such that $(j_k,x_{\ell_{j_k}},r_{j_k})$ has been queried by $\mathcal{A}^\mathsf{RO}$ on input $X$ before the beginning of round $k$. 
\begin{defn}
Let $\mathcal{A}$ be some deterministic massively parallel computation. Let $X\in\{0,1\}^{uv}$ be some input, and $\mathsf{RO}:\{0,1\}^n\to\{0,1\}^n$ be some oracle s.t. running $\mathcal{A}^\mathsf{RO}$ on input $X$, $\overline{E^{(k)}}$ happens.
Then, for any $a_1,...,a_{\log^2 w}\in [v]^{\log^2 w}$, round $k$, let $\mathsf{RO}_{a_1,...,a_{\log^2 w}}^{(k)}$ be the oracle constructed by the following procedure. 
\end{defn}
\ifnum\submission=0
\begin{enumerate}
    \item $\mathsf{RO}_{a_1,...,a_{\log^2 w}}^{(k)}\leftarrow \mathsf{RO}$.
    \item Let $a_0=\ell_{j_k}$ and $r_{j_k}'=r_{j_k}$. For $t=1,...,\log^2 w$, 
    \begin{align*}
    \mathsf{RO}_{a_1,...,a_{\log^2 w}}^{(k)}&(j_k+t-1,x_{a_{t-1}},r_{j_k+t-1}',0^*)\leftarrow (a_t,r_{j_k+t}',z_{j_k+t}'),
    \end{align*}
    where $(\ell_{j_k+t}',r_{j_k+t}',z_{j_k+t}')\coloneqq \mathsf{RO}(j_k+t-1,x_{a_{t-1}},r_{j_k+t-1}',0^*)$.
\end{enumerate}
\else
\begin{enumerate}
    \item $\mathsf{RO}_{a_1,...,a_{\log^2 w}}^{(k)}\leftarrow \mathsf{RO}$.
    \item Let $a_0=\ell_{j_k}$ and $r_{j_k}'=r_{j_k}$. For $t=1,...,\log^2 w$, 
    \begin{align*}
    \mathsf{RO}_{a_1,...,a_{\log^2 w}}^{(k)}&(j_k+t-1,x_{a_{t-1}},r_{j_k+t-1}',0^*)\\
    \leftarrow &(a_t,r_{j_k+t}',z_{j_k+t}'),
    \end{align*}
    where $(\ell_{j_k+t}',r_{j_k+t}',z_{j_k+t}')\coloneqq \mathsf{RO}(j_k+t-1,x_{a_{t-1}},r_{j_k+t-1}',0^*)$.
\end{enumerate}
\fi
We note that given $(\mathsf{RO},X)$ s.t. $\overline{E^{(k)}}$ happens, all $\mathsf{RO}_{a_1,...,a_{\log^2 w}}^{(k)}$ have the same $j_k$ and thus $\mathsf{RO}_{a_1,...,a_{\log^2 w}}^{(k)}$ is well-defined. 
\begin{defn}\label{set}
Let $\mathcal{A}$ be some deterministic massively parallel computation with $m$ machines. Let $\mathsf{RO}:\{0,1\}^n\to\{0,1\}^n$ be some oracle and $X\in\{0,1\}^{uv}$ be an input s.t. running $\mathcal{A}^\mathsf{RO}$ on input $X$, $\overline{E^{(k)}}$ happens. Further, let $i\in [m]$ be the index of some machine in $\mathcal{A}$, and $k\geq 0$ be some integer. Let the set $B_i^{(k)}\subseteq [v]$ be s.t. $a\in B_i^{(k)}$ if there is a sequence $a_1,...,a_{\log^2 w}\in [v]^{\log^2 w}$ and $b\in [\log^2 w]$ such that $a_b=a$ and running $\mathcal{A}$ with oracle access to $\mathsf{RO}_{a_1,...,a_{\log^2 w}}^{(k)}$ on input $X$, machine $i$ queries $(j_k+b,x_{a},r_{j_k+b}',0^*)$ in round $k$. 
\end{defn}

Recall that at the beginning of a round, each machine receives a state of size $s$ which may depend on all the previous queries and the input $X=x_1,...,x_v$. Lemma~\ref{nojump2} helps us rule out the possibility that the given state depends on any element in $V^{(k)}$ before the beginning of round $k$, in which case our encoding scheme cannot work. We have the following lemma which helps us bound the size of $B_i^{(k)}$.
\begin{lma}\label{bound}
Let $k$ be an integer, $E^{(k)}$ be the event defined in Lemma~\ref{nojump2} and $\overline{E^{(k)}}$ be its negation. And suppose $u\geq (\log^2 w+2)\log v+\log q$. Let $\mathcal{A}$ be a deterministic massively parallel computation with $m$ machines, local memory of size $s$ and the number of queries $q$ computing $\Line_{n,w,u,v}$. Then, for any machine $i$, any round $k$, we have 
\[
    \Pr_{(\mathsf{RO},X)}\left[\abs{B_i^{(k)}}> h \land \overline{E^{(k)}}\right]\leq 2^{-(u-(\log^2 w+2)\log v-\log q)},
\]
where $h=\frac{s}{u-(\log^2 w+2)\log v-\log q}+1$ and $\mathsf{RO},X$ are uniformly distributed.
\end{lma}
\begin{proof}
We show that the fraction of $(\mathsf{RO}, X)$ s.t. $\abs{B_i^{(k)}}> h$ and $E^{(k)}$ does not happen cannot be too large. We first construct an encoding scheme that ``compresses'' all the $(\mathsf{RO},X)$ such that $\abs{B_i^{(k)}}> h$ and $E^{(k)}$ does not happen. In the following, we consider $\mathcal{A}$ running until the end of round $k$. 
\begin{cla}\label{enc2}
If
\[
    \Pr_{(\mathsf{RO},X)}\left[\abs{B_i^{(k)}}> h\land \overline{E^{(k)}}\right]=\epsilon,
\]
then there is a set $F$ of $(\mathsf{RO},X)$ such that $\abs{F}\geq \epsilon 2^{n2^n+uv}$ and a deterministic encoding scheme $(\Enc,\Dec)$ such that for any $(\mathsf{RO},X)\in F$, both of the following hold
\begin{enumerate}
    \item $\Dec(\Enc(\mathsf{RO},X))=(\mathsf{RO},X)$
    \item $\abs{\Enc(\mathsf{RO},X)}\leq s+h((\log^2 w+2)\log v+\log q)+(v-h)u+n2^n$.
\end{enumerate}
\end{cla}
\begin{proof}
We consider all the computation done by $\mathcal{A}$ before the beginning of round $k$ as $\mathcal{A}_1$ and the output of $\mathcal{A}_1$ is the memory state given to machine $i$ as input at the beginning of round $k$. We also consider the computation done by machine $i$ in round $k$ as $\mathcal{A}_2$ and it outputs the set of queries, and the corresponding answer set. Since for each $(\mathsf{RO},X)\in F$, when running on $(\mathsf{RO},X)$, $E^{(k)}$ does not happen, we assure that $\mathcal{A}_1$ never queries any element in $V^{(k)}$ and hence, $\mathcal{A}_2$ can not tell whether we replace $\mathrm{RO}$ with $\mathrm{RO}_{a_1,...,a_{\log^2 w}}$.

Now we are able to describe our encoding scheme.

$\Enc(\mathsf{RO},X):$
\begin{enumerate}
    \item Add the entire $\mathsf{RO}$ to our encoding.
    \item Run $\mathcal{A}_1(X)$ with oracle access to $\mathsf{RO}$. Denote its output as $M$ and add $M$ to our encoding. Note that $|M|=s$.
    \item For any $a_1,...,a_{\log^2 w}\in [v]^{\log^2 w}$, run $\mathcal{A}_2(M)$ with oracle access to $\mathsf{RO}_{a_1,...,a_{\log^2 w}}$. This can be done by examining the queries of $\mathcal{A}_2$ and providing a revised answer $(a_t,r_{j_k+t}',z_{j_k+t}')$ to $\mathcal{A}_2$ if it makes a corresponding query.
    \begin{itemize}
        \item Let $q_t=(j_k+t,x_{a_t},r_{j_k+t}',0^*)$ for every $t\in [\log^2 w]$. On $a_1,...,a_{\log^2 w}$, denote $Q_{a_1,...,a_{\log^2 w}}$ as the set of $q_t$ such that $q_t$ is queried by $\mathcal{A}_2$ and the corresponding $a_t$ hasn't been recorded before. If $Q_{a_1,...,a_{\log^2 w}}\neq \phi$, add $a_1,...,a_{\log^2 w}$, $\abs{Q_{a_1,...,a_{\log^2 w}}}$, the index of each \\$q_t\in Q_{a_1,....,a_{\log^2 w}}$ and the corresponding $a_t$ of all the $q_t \in Q_{a_1,...,a_{\log^2 w}}$ to the encoding.
    \end{itemize}
    \item For each $x\in X$ that is not contained in any query of $\mathcal{A}_2$ in the above process, add $x$ to our encoding. Denote it as $X'$.
\end{enumerate}
Denote the entire encoding as $msg$.
Note that in the third step, by our construction, the union of $Q_{a_1,...,a_{\log^2 w}}$ is exactly $B_i^{(k)}$. Therefore, the size of encoding added in this step is at most 
\[
    |B_i^{(k)}|((\log^2w+1)\log v+\log q+\log |B_i^{(k)}|).
\]
    Thus, given that $v\geq \abs{B_i^{(k)}}>h$ and $u\geq (\log^2 w+2)\log v+\log q$, the total size of encoding is at most
\[
    s+h((\log^2 w+2)\log v+\log q)+(v-h)u+n2^n.
\]
$\Dec(msg):$
\begin{enumerate}
    \item Construct $\mathsf{RO}$ from the first part of $msg$.
    \item For each $a_1,...,a_{\log^2 w}$ in the $msg$, run $\mathcal{A}_2(M)$ with oracle access to $\mathsf{RO}_{a_1,...,a_{\log^2 w}}$.
    \begin{itemize}
        \item Read $\abs{Q_{a_1,...,a_{\log^2 w}}}$ to see how many $x_i$ to recover using the current $a_1,...a_{\log^2 w}$.
        \item Read each index of query $q_t$ and corresponding $a_t$ in the $msg$ and find the corresponding query from the queries of $\mathcal{A}_2$ to recover $x_{a_t}$.
    \end{itemize}
    \item The remaining of $X$ can be reconstructed using $X'$.
\end{enumerate}

The correctness follows clearly.
\end{proof}
The following claim shows the information-theoretic limit of any deterministic encoding scheme with perfect correctness.
\begin{cla}\label{limit2}
For any deterministic encoding scheme $(\Enc,\Dec)$ such that $\forall m\in M$, $\Dec(\Enc(m))=m$, we have
\[
     \max_{m} \abs{\Enc(m)}\geq \log \abs{M}-1.
\]
\end{cla}
\begin{proof}
Suppose $\max_{m}\abs{\Enc(m)}=t$. Then the number of possible codewords is 
\[
    \sum_{i=0}^t 2^{t-i}\leq 2^{t+1}.
\]
To have a one-to-one mapping, we have 
$
    2^{t+1}\geq \abs{M}$, and thus $t\geq \log \abs{M} -1$.
\end{proof}
Suppose 
\[
    \Pr_{(\mathsf{RO},X)}\left[\abs{B_i^{(k)}}> h\land \overline{E^{(k)}}\right]=\epsilon.
\]
Then, by Claim~\ref{enc2}, there is a set $F$ of $(\mathsf{RO},X)$ such that $\abs{F}\geq \epsilon 2^{n2^n+uv}$ and a deterministic encoding scheme $(\Enc,\Dec)$ such that 

\ifnum\submission=0
\begin{equation}\label{upbound}
    \abs{\Enc(\mathsf{RO},X)}\leq s+h((\log^2 w+2)\log v+\log q)+(v-h)u+n2^n
\end{equation}

\else

\begin{align}\label{upbound}
\begin{split}
    \abs{\Enc(\mathsf{RO},X)}
    \leq s+h((\log^2 w&+2)\log v+\log q)+(v-h)u+n2^n
\end{split}
\end{align}
\fi
and $\Dec(\Enc(\mathsf{RO},X))=(\mathsf{RO},X)$ for any $(\mathsf{RO},X)\in F$. 
On the other hand, by Claim~\ref{limit2}, for any deterministic encoding scheme $(\Enc', \Dec')$ such that 
\[
    \Dec'(\Enc'(m))=m, \forall m\in F,
\]
we have
\begin{equation}\label{lowbound}
    \max_m \abs{\Enc'(m)}\geq \log \abs{F}-1\geq n2^n+uv+\log \epsilon-1.
\end{equation}
Combining Equation \ref{upbound}, Equation \ref{lowbound} and that \[h=\frac{s}{u-(\log^2 w+2)\log v-\log q}+1,\] the lemma follows.
\end{proof}
Now we are in a position to prove Lemma~\ref{lmaline}. For each round $k\geq 0$, let \[
    C^{(k)}=\{(i,x_{\ell_i},r_i)\mid k\log^2 w< i\leq w\}.
\] For each round $k$, let $Q^{(\leq k)}$ be the set of queries done by all machines until the end of round $k$, $Q^{(k)}$ be the set of queries done by all machines in round $k$, and $Q^{(k)}_i$ be the set of queries done by machine $i$ in round $k$.
\begin{proof}[Proof of Lemma~\ref{lmaline}]
We prove Lemma~\ref{lmaline} by showing the following claim.
\begin{cla}
For any deterministic massively parallel computation with $m$ machines, local memory of size $s$ and the number of queries $q$ computing $\Line_{n,w,u,v}$ and running until the end of round $0\leq k<\frac{w}{\log^2 w}-1$,
\ifnum\submission=0
\begin{align*}
    \Pr_{(\mathsf{RO},X)}\left[\abs{Q^{(\leq k)}\cap C^{(k+1)}}>0\right]
    \leq (k+1)m\left(\left(\frac{h}{v}\right)^{\log^2 w}+wv^{\log^2w}q2^{-u}+2^{-(n-(\log^2 w+2)\log v-\log q)}\right),
\end{align*}

\else
{\small
\begin{align*}
    &\Pr_{(\mathsf{RO},X)}\left[\abs{Q^{(\leq k)}\cap C^{(k+1)}}>0\right]\\
    &\leq (k+1)\cdot m \cdot \left(\left(\frac{h}{v}\right)^{\log^2 w}+wqv^{\log^2w}2^{-u}+2^{-(u-(\log^2 w+2)\log v-\log q)}\right)
\end{align*}
}
\fi
, where $h=\frac{s}{u-(\log^2 w+2)\log v-\log q}+1$.
\end{cla}
\begin{proof}
We prove this by induction. 
For the base case, $k=0$, we have, for any machine $i$,
\ifnum\submission=0
{\small
\begin{align*}
   \Pr_{(\mathsf{RO},X)}\left[\abs{Q^{(0)}_i\cap C^{(1)}}>0\land \overline{E^{(0)}}\right]
    &\leq \Pr\left[\abs{Q^{(0)}_i\cap C^{(0)}}>\log^2 w\land \abs{B_i^{(0)}}\leq h\land \overline{E^{(0)}}\right]+\Pr\left[\abs{B_i^{(0)}}> h\land \overline{E^{(0)}}\right]\\
    &\leq \Pr\left[\abs{Q^{(0)}_i\cap C^{(0)}}>\log^2 w\land \abs{B_i^{(0)}}\leq h \middle |  \overline{E^{(0)}}\right]+\Pr\left[\abs{B_i^{(0)}}> h\land \overline{E^{(0)}}\right]\\
    &\leq \left(\frac{h}{v}\right)^{\log^2 w}+2^{-(u-(\log^2 w+2)\log v-\log q)},
\end{align*}
}
\else
{\small
\begin{align*}
   &\Pr_{(\mathsf{RO},X)}\left[\abs{Q^{(0)}_i\cap C^{(1)}}>0\land \overline{E^{(0)}}\right]\\
    &\leq \Pr\left[\abs{Q^{(0)}_i\cap C^{(0)}}>\log^2 w\land \abs{B_i^{(0)}}\leq h\land \overline{E^{(0)}}\right]+\Pr\left[\abs{B_i^{(0)}}> h\land \overline{E^{(0)}}\right]\\
    &\leq \Pr\left[\abs{B_i^{(0)}}> h\land \overline{E^{(0)}}\right]+\Pr\left[\abs{Q^{(0)}_i\cap C^{(0)}}>\log^2 w\land \abs{B_i^{(0)}}\leq h \middle | \overline{E^{(0)}}\right]\\
    &\leq \left(\frac{h}{v}\right)^{\log^2 w}+2^{-(u-(\log^2 w+2)\log v-\log q)}
\end{align*}
}\fi, where the last inequality follows from Lemma~\ref{bound} and the fact that $\abs{B_i^{(0)}}$ is the number of $x$ s.t. machine $i$ is able to output in round $0$, given any possible sequences of $\log^2 w$ consecutive $\ell$'s. \\
Thus, by a union bound, we obtain
\ifnum\submission=0
\begin{align*}
    \Pr_{(\mathsf{RO},X)}\left[\abs{Q^{(\leq 0)}\cap C^{(1)}}>0\right]
    &\leq m \Pr\left[\abs{Q^{(0)}_i\cap C^{(1)}}>0\land \overline{E^{(0)}}\right]+\Pr\left[E^{(0)}\right]\\
    &\leq m\left(\left(\frac{h}{v}\right)^{\log^2 w}+wv^{\log^2w}q2^{-u}+2^{-(u-(\log^2 w+2)\log v-\log q)}\right).
\end{align*}
\else
{\small
\begin{align*}
    &\Pr_{(\mathsf{RO},X)}\left[\abs{Q^{(\leq 0)}\cap C^{(1)}}>0\right]\\
    &\leq m \Pr\left[\abs{Q^{(0)}_i\cap C^{(1)}}>0\land \overline{E^{(0)}}\right] +\Pr\left[E^{(0)}\right]\\
    &\leq m\left(\left(\frac{h}{v}\right)^{\log^2 w} + wv^{\log^2w}q2^{-u}+2^{-(u-(\log^2 w+2)\log v-\log q)}\right).
\end{align*}
}
\fi

Now assume that for round $k-1$, the claim holds. Hence, we have
\ifnum\submission=0
{\small
\begin{equation}\label{equ:last_equ_1}
\begin{split}
    \Pr\left[\abs{Q^{(\leq k)}\cap C^{(k+1)}}>0\right]
    \leq &\Pr\left[\abs{Q^{(\leq k-1)}\cap C^{(k)}}>0\right]
    +\Pr\left[\abs{Q^{(\leq k)}\cap C^{(k+1)}}>0 \land \abs{Q^{(\leq k-1)}\cap C^{(k)}}=0\right]\\
    \leq & \Pr\left[\abs{Q^{(\leq k-1)}\cap C^{(k)}}>0\right]\\
         & +\Pr\left[\abs{Q^{(k)}\cap C^{(k+1)}}>0\land \abs{Q^{(\leq k-1)}\cap C^{(k)}}=0\land \overline{E^{(k)}}\right]
         \\ & +  \Pr\left[E^{(k)}\right]\\
\end{split}\end{equation}
Notice that 
\begin{equation}\label{equ:last_equ_2}\begin{split}
    & \Pr\left[\abs{Q^{(k)}\cap C^{(k+1)}}>0\land \abs{Q^{(\leq k-1)}\cap C^{(k)}}=0\land \overline{E^{(k)}}\right]\\
    \leq&  m \Pr\left[\abs{Q_i^{(k)}\cap C^{(k+1)}}>0\land \abs{Q^{(\leq k-1)}\cap C^{(k)}}=0\land \overline{E^{(k)}}\right]\\
    \leq & m \Pr\left[\abs{Q_i^{(k)}\cap C^{(k)}}>\log^2 w \land \abs{Q^{(\leq k-1)}\cap C^{(k)}}=0 \land \abs{B_i^{(k)}}\leq h\middle | \overline{E^{(k)}}\right]\\
    & +  m\cdot  \Pr\left[\abs{B_i^{(k)}}> h\land \overline{E^{(k)}}\right]\\
     \leq & m\cdot \left(\left(\frac{h}{v}\right)^{\log^2 w}+2^{-(u-(\log^2 w+2)\log v-\log q)}\right)
\end{split}\end{equation}
}
\else
{\small
\begin{equation}\label{equ:last_equ_1}
\begin{split}
    & \Pr\left[\abs{Q^{(\leq k)}\cap C^{(k+1)}}>0\right]\\
    \leq &\Pr\left[\abs{Q^{(\leq k-1)}\cap C^{(k)}}>0\right] \\ 
    & +\Pr\left[\abs{Q^{(\leq k)}\cap C^{(k+1)}}>0 \land \abs{Q^{(\leq k-1)}\cap C^{(k)}}=0\right]\\
    \leq & \Pr\left[\abs{Q^{(\leq k-1)}\cap C^{(k)}}>0\right]\\
         & +\Pr\left[\abs{Q^{(k)}\cap C^{(k+1)}}>0\land \abs{Q^{(\leq k-1)}\cap C^{(k)}}=0\land \overline{E^{(k)}}\right]
         \\ & +  \Pr\left[E^{(k)}\right]\\
\end{split}\end{equation}
Notice that 
\begin{equation}\label{equ:last_equ_2}\begin{split}
    & \Pr\left[\abs{Q^{(k)}\cap C^{(k+1)}}>0\land \abs{Q^{(\leq k-1)}\cap C^{(k)}}=0\land \overline{E^{(k)}}\right]\\
    \leq&  m \Pr\left[\abs{Q_i^{(k)}\cap C^{(k+1)}}>0\land \abs{Q^{(\leq k-1)}\cap C^{(k)}}=0\land \overline{E^{(k)}}\right]\\
    \leq & m \Pr\left[\abs{Q_i^{(k)}\cap C^{(k)}}>\log^2 w \land \abs{Q^{(\leq k-1)}\cap C^{(k)}}=0 \land \abs{B_i^{(k)}}\leq h\middle | \overline{E^{(k)}}\right]\\
    & +  m\cdot  \Pr\left[\abs{B_i^{(k)}}> h\land \overline{E^{(k)}}\right]\\
     \leq & m\cdot \left(\left(\frac{h}{v}\right)^{\log^2 w}+2^{-(u-(\log^2 w+2)\log v-\log q)}\right)
\end{split}\end{equation}
}\fi
, where the last inequality holds since we think of it as first fixing all the oracle answers on the queries before the beginning of the $k$th round and the remaining part of the oracle remains uniformly random. Now all the $\mathsf{RO}_{a_1,...,a_{\log^2 w}}^{(k)}$ and hence $B_i^{(k)}$ are well-defined and we are able to apply Lemma~\ref{bound}. The second probability bound follows from the definition of $|B_i^{(k)}|$ and the fact that oracle answers in $C^k$ remains uniformly random.

By Equation~\ref{equ:last_equ_1}, Equation~\ref{equ:last_equ_2}, and the inductive hypothesis, we have 
\[\begin{split}
 \Pr\left[\abs{Q^{(\leq k)}\cap C^{(k+1)}}>0\right]
\leq & (k+1)m\left(\left(\frac{h}{v}\right)^{\log^2 w}+wqv^{\log^2w}2^{-u} +2^{-(u-(\log^2 w+2)\log v-\log q)}\right)
\end{split}\]
\end{proof}
Let $\textbf{Success}$ be the event that $\mathcal{A}$ successfully compute $\Line_{n,w,u,v}^{\mathsf{RO}}$ in $\frac{w}{\log^2 w}$ round. To compute $\Line_{n,w,u,v}^\mathsf{RO}$, the algorithm must reach $(w,x_{\ell_{w}},r_{w})$. However, $(w,x_{\ell_{w}},r_{w})\in C^{(\frac{w}{\log^2 w}-1)}$ and hence, by our claim,
\ifnum\submission=0
\begin{align*}
    \Pr_{(\mathsf{RO},X)}[\textbf{Success}]
    &\leq \Pr_{(\mathsf{RO},X)}\left[\left|Q^{(\leq \frac{w}{\log^2 w}-2)}\cap C^{(\frac{w}{\log^2 w}-1)}\right|>0\right]\\
    &\leq \frac{w}{\log^2 w}m\left(\left(\frac{h}{v}\right)^{\log^2 w}+v^{\log^2w}q2^{-u}+2^{-(u-(\log^2 w+2)\log v-\log q)}\right).
\end{align*}
\else
{\small
\begin{align*}
    &\Pr_{(\mathsf{RO},X)}[\textbf{Success}]\\
    &\leq \Pr_{(\mathsf{RO},X)}\left[\left|Q^{(\leq \frac{w}{\log^2 w}-2)}\cap C^{(\frac{w}{\log^2 w}-1)}\right|>0\right]\\
    &\leq \frac{mw}{\log^2 w}\cdot
    \left(\left(\frac{h}{v}\right)^{\log^2 w}+v^{\log^2w}q2^{-u}+2^{-(u-(\log^2 w+2)\log v-\log q)}\right).
\end{align*}
}
\fi
As $n$ becomes sufficiently large, by our parameters setting, the success probability becomes sufficiently small.
\end{proof}

\bibliographystyle{plain}
\bibliography{ref}

\appendix
\section*{Appendix}
\section{A Warm-up Result}
\label{warmup}
In this section, we present a simpler construction to give some intuition on how our argument works.
We first state our theorem in random oracle model. 

\begin{thm}\label{simline}
There exists a universal constant $c>1$ such that for any sufficiently large $n>0$, let $\mathsf{RO}:\{0,1\}^n\to\{0,1\}^n$ be a random oracle, and running time $S\leq T<2^{O(n)}$, there is an oracle function $f:\{0,1\}^S\to \{0,1\}^n$ such that it can be computed in time $O(T\cdot n)$ using memory size $O(S)$ by a RAM algorithm in random oracle model. On the other hand, let $\mathcal{A}^\mathsf{RO}$ be a randomized massively parallel computation with $m<2^{O(n)}$ machines, local memory of size $s\leq S/c$ and the number of local queries $q<2^{O(n)}$ to random oracle per round. Then, in random oracle model, $\mathcal{A}^\mathsf{RO}$ needs at least $R\geq \Omega(\frac{T}{s})$ rounds to compute the function in average case. 
\end{thm}

Note that Theorem~\ref{simline} is information-theoretic in the sense that even if we allow each machine to do arbitrary computation in each round, our lower bound result still holds. 

The function we consider in Theorem~\ref{simline} is $\SimLine^\mathsf{RO}_{n,w,u,v}:\{0,1\}^{uv}\to\{0,1\}^n$ defined as follows: Given input $x=x_1,x_2,...,x_v$ such that $x_i\in\{0,1\}^u$ for all $i\in [v]$ and a random oracle $\mathsf{RO}:\{0,1\}^{n}\to\{0,1\}^{n}$, let $r_1=0^u$
and 
\begin{align*}
    (r_{i+1},z_{i+1}) \coloneqq \mathsf{RO}(x_{i\Mod v},r_i,0^*),\quad \forall i\in[w],
\end{align*}
the output of $\SimLine^\mathsf{RO}_{n,w, u, v}(x)$ is defined as the answer to the last query, $(r_{w+1},z_{w+1})$. 

Given the parameters $S,T$ in Theorem~\ref{simline}, we set the parameters of $\SimLine^\mathsf{RO}$ as $w=T$, $v=S/u$ and $u=n/3$. One can observe that the obvious RAM algorithm which queries $(x_i,r_i)$ one by one already meets the performance on memory size and running time stated in Theorem~\ref{simline}. Thus, we focus on the lower bound of massively parallel computation. As the standard observation in Remark~\ref{ran}, without loss of generality, we assume the MPC computation is deterministic. In particular, we can conclude Theorem~\ref{simline} from the following lemma.

\begin{lma}\label{lmasimline}
There exists a universal constant $c>1$ such that for any sufficiently large $n>0$, let $\mathsf{RO}:\{0,1\}^n\to\{0,1\}^n$ be a random oracle and for any $n\leq S<2^{O(n)}$ and $S\leq T<2^{O(n)}$, consider the oracle function $\SimLine^\mathsf{RO}_{n,w,u,v}:\{0,1\}^{uv}\to\{0,1\}^n$ where $w=T$, $v=S/u$ and $u=n/3$. Let $\mathcal{A}^\mathsf{RO}$ be a deterministic massively parallel computation with $m<2^{O(n)}$ machines, local memory of size $s\leq S/c$ and the number of local queries $q<2^{O(n)}$ to random oracle per round. Then, in random oracle model, $\mathcal{A}^\mathsf{RO}$ needs at least $R\geq \frac{w}{s/(u-\log q-\log v)+1}\geq \Omega(\frac{T}{s})$ rounds to compute $\SimLine^\mathsf{RO}_{n,w,u,v}$ in random oracle model.
\end{lma}

  We first give the intuition of our lower bound. Consider a local algorithm $\mathcal{A}$ at the beginning of certain round and the set of $x_i$'s stored in its local memory. Suppose the size of this set is $r$. Then, obviously, we can bound the number of correct queries of $\mathcal{A}$ by $r$. To formalize this, we need a definition that effectively captures the set of $x_i$ mentioned above. In particular, we define the set $B$ to be the set containing those $x$’s that appear in the queries of the algorithm. 
  
  Now we give a high-level overview of our technique. Our argument centers around an encoding scheme that encodes the random oracle $\mathsf{RO}$ and the input $X$. The main idea of this encoding scheme is to retrieve $x_i$'s from the queries of local machine. In particular, the encoding contains the local memory, random oracle, the $x_i$'s that is not retrievable from the queries, and some auxiliary information indicating where to retrieve $x_i$'s from the queries. To decode, first run the local algorithm on local memory with access to the stored oracle to obtain the set of queries, then, use the auxiliary information to retrieve $x_i$'s contained in the queries, and combine them with the remaining $x_i$'s. Since the size of local memory is small, the encoding scheme will go beyond the information-theoretic limit if the set of queries contains many $x_i$'s, which leads to a contradiction. If we use the local algorithm in this way, we can bound the size of intersection between the set of queries and the set of correct entries in $\SimLine^\mathsf{RO}$. This allows us to bound the number of steps a machine can advance in a round by the maximum number of $x_i$'s it can store in local memory. 

Let $h=\frac{s}{u-\log q-\log v}+1$. To simplify the notation, we assume $\frac{w}{h}$ is an integer. For each $0\leq j\leq w/h-1$, let 
\[
C_j=\left\{(x_{i\Mod v},r_i)\mid jh+1\leq i\leq \min(jh+v,w)\right\}
\]
be the set containing no duplicate $x_i$. 
\begin{lma}\label{smallint}
Given $0\leq j\leq w/v-1$, a subset $C\subseteq C_j$ and a pair of deterministic algorithms $(\mathcal{A}_1,\mathcal{A}_2)$ such that $\mathcal{A}_1$ with oracle access to $\mathsf{RO}$ is given $vu$-bit $X=x_0,x_1,...,x_{v-1}$ as input and outputs $s$-bit state $M$ and $\mathcal{A}_2$ has oracle access to $\mathsf{RO}$ and given $M$ as input, outputs a set of its queries $Q$ to the oracle $\mathsf{RO}$ and a set of corresponding answers $A$, where $\abs{Q}=\abs{A}=q$, we have, for any $\alpha>0$,
\ifnum\submission=0
\[
    \Pr_{(\mathsf{RO},X)}\left[\abs{Q\cap C}\geq \alpha: M\leftarrow \mathcal{A}_1^\mathsf{RO}(X), Q,A\leftarrow \mathcal{A}_2^\mathsf{RO}(M)\right]\leq 2^{-(\alpha(u-\log q-\log v)-s-1)}
\]
\else
\begin{align*}
    \Pr_{(\mathsf{RO},X)}&\left[\abs{Q\cap C}\geq \alpha: M\leftarrow \mathcal{A}_1^\mathsf{RO}(X), Q,A\leftarrow \mathcal{A}_2^\mathsf{RO}(M)\right] \\
    &\leq 2^{-(\alpha(u-\log q-\log v)-s-1)},
\end{align*}
where $\mathsf{RO}$ and $X$ are uniformly distributed.
\fi
\end{lma}
\begin{proof}

\begin{cla}\label{enc}
If
\ifnum\submission=0
\[
    \Pr_{(\mathsf{RO},X)}\left[\abs{Q\cap C}\geq \alpha: M\leftarrow \mathcal{A}_1^\mathsf{RO}(X), Q,A\leftarrow \mathcal{A}_2^\mathsf{RO}(M)\right]=\epsilon,
\]
\else
{\small
\[
    \Pr_{(\mathsf{RO},X)}\left[\abs{Q\cap C}\geq \alpha: M\leftarrow \mathcal{A}_1^\mathsf{RO}(X), Q,A\leftarrow \mathcal{A}_2^\mathsf{RO}(M)\right]=\epsilon,
\]
}
\fi
then there is a set $F\subseteq\{(\mathsf{RO},X)\mid \mathsf{RO}:\{0,1\}^n\to\{0,1\}^n, X\in\{0,1\}^{uv}\}$ such that $\abs{F}\geq \epsilon 2^{n2^n+uv}$ and a deterministic encoding scheme $(\Enc,\Dec)$ with black-box access to $\mathcal{A}_1$ and $\mathcal{A}_2$ such that for any $(\mathsf{RO},X)\in F$, both of the following hold
\begin{enumerate}
    \item $\Dec(\Enc(\mathsf{RO},X))=(\mathsf{RO},X)$
    \item $\abs{\Enc(\mathsf{RO},X)}\leq s+\alpha(\log q+\log v)+(v-\alpha)u+2^nn$
\end{enumerate}

\end{cla}
\begin{proof}
Since 
\ifnum\submission=0
\[
    \Pr_{(\mathsf{RO},X)}\left[\abs{Q\cap C}\geq \alpha: M\leftarrow \mathcal{A}_1^\mathsf{RO}(X), Q,A\leftarrow \mathcal{A}_2^\mathsf{RO}(M)\right]=\epsilon
\]
\else
{\small
\[
    \Pr_{(\mathsf{RO},X)}\left[\abs{Q\cap C}\geq \alpha: M\leftarrow \mathcal{A}_1^\mathsf{RO}(X), Q,A\leftarrow \mathcal{A}_2^\mathsf{RO}(M)\right]=\epsilon
\]
}
\fi
and $\mathcal{A}_1$, $\mathcal{A}_2$ are deterministic, there is a set 
\[
    F\subseteq \{(\mathsf{RO},X) \mid \mathsf{RO}:\{0,1\}^n\to \{0,1\}^n,X\in\{0,1\}^{uv}\}
\]
such that 
\[
    \abs{F}\geq \epsilon 2^{n2^n+uv}
\]
and 
\ifnum\submission=0
\[
    \Pr\left[\abs{Q\cap C}\geq \alpha: M\leftarrow \mathcal{A}_1^\mathsf{RO}(X), Q,A\leftarrow \mathcal{A}_2^\mathsf{RO}(M)\right]=1, \forall (\mathsf{RO},X)\in F.
\]
\else
\[
    \Pr\left[\abs{Q\cap C}\geq \alpha: M\leftarrow \mathcal{A}_1^\mathsf{RO}(X), Q,A\leftarrow \mathcal{A}_2^\mathsf{RO}(M)\right]=1
\]
, $\forall (\mathsf{RO},X)\in F$.
\fi

We describe our encoding scheme that encodes all $(\mathsf{RO},X)\in F$.

$\Enc(\mathsf{RO},X):$
\begin{enumerate}
    \item Add entire oracle $\mathsf{RO}$ to our encoding.
    \item $M\leftarrow \mathcal{A}_1^\mathsf{RO}(X)$, add $M$ to our encoding.
    \item Run $Q,A\leftarrow \mathcal{A}_2^\mathsf{RO}(M)$.
    \item For each $c_i\in C$, if $c_i=(x,r)\in Q$, then record index of this query, $p_i$, and its index in $X$, $I_i$. Let $P=\{(p_i,I_i)\mid c_i\in C\}$ and add $P$ to our encoding. Note that $p_i$ takes $\log q$ bits, $I_i$ takes $\log v$ bits and $|P|\geq \alpha$.
    \item For each $x\in X$ but $x\notin C$, add $x$ to our encoding (in the order of $\SimLine^\mathsf{RO}$). Denote it as $X'$.
\end{enumerate}
As long as $u\geq \log q+\log v$, the encoding takes size at most
\[
    s+\alpha(\log q+\log v)+(v-\alpha)u+2^nn.
\]
$\Dec(\mathsf{RO},M,P,X'):$
\begin{enumerate}
    \item Run $\mathcal{A}_2^\mathsf{RO}(M)$. 
    \item Use the recorded position in $P$ to recover those recorded $x$.
    \item Use $X'$ to recover the remaining $x\in X$
\end{enumerate}

Since we answer the queries of $\mathcal{A}_2$ using the same oracle and $\mathcal{A}_2$ is deterministic, the queries of $\mathcal{A}_2$ when decoding are the same as the ones when encoding. Hence, we can correctly construct some of $x$ from the positions recorded in $P$. This completes the proof.
\end{proof}

\begin{cla}\label{limit}
For any deterministic encoding scheme $(\Enc,\Dec)$ such that $\Dec(\Enc(m))=m$, $\forall m\in M$, we have
\[
     \max_{m} \abs{\Enc(m)}\geq \log \abs{M}-1.
\]
\end{cla}
\begin{proof}
Suppose $\max_{m}\abs{\Enc(m)}=t$. Then the number of possible codewords is 
\[
    \sum_{i=0}^t 2^{t-i}\leq 2^{t+1}.
\]
To have a one-to-one mapping, the following must be true.
\[
    2^{t+1}\geq \abs{M}
\]
And hence $t\geq \log \abs{M} -1$.
\end{proof}
Now we are able to prove Lemma~\ref{smallint}.\\
Suppose 
\[\begin{split}
    & \Pr_{(\mathsf{RO},X)}\left[\abs{Q\cap C}\geq \alpha:M\leftarrow \mathcal{A}_1^\mathsf{RO}(X),Q,A\leftarrow \mathcal{A}_2^\mathsf{RO}(M)\right] \\ & =\epsilon.
    \end{split}
\]
Then by Claim~\ref{enc}, there is a set 
\[
    F\subseteq \{(\mathsf{RO},X)\mid \mathsf{RO}:\{0,1\}^n\to \{0,1\}^n, X\in\{0,1\}^{uv}\}
\]
such that  $\abs{F}\geq \epsilon 2^{n2^n+uv}$ and a deterministic encoding scheme $(\Enc,\Dec)$ with blackbox access to $\mathcal{A}_1$ and $\mathcal{A}_2$ such that 
                                              $\abs{\Enc(\mathsf{RO},X)}\leq s+\alpha(\log q+\log v)+(v-\alpha)u+2^nn$ for any $(\mathsf{RO},X)\in F$. 
On the other hand, by Claim~\ref{limit}, for any deterministic encoding scheme $(\Enc', \Dec')$ such that 
\[
    \Dec'(\Enc'(m))=m, \forall m\in F,
\]
we have
\[
    \max_m \abs{\Enc'(m)}\geq \log \abs{F}-1\geq n2^n+uv+\log \epsilon-1.
\]
Combining these, the lemma follows.
\end{proof}
 Let $Q_i^{(k)}$ be the set of queries done by machine $i$ in round $k$. We can apply Lemma~\ref{smallint} to show the following lemma in massively parallel computation model.
\begin{lma}\label{mrsmallint}
For any deterministic massively parallel computation $\mathcal{A}$ with $m$ machines, local memory of size $s$ and the number of queries $q$ computing $\SimLine_{n,w,u,v}$, for any machine $i$, any round $k\geq 0$, and any subset $C\subseteq C_j$ where $w/h-1\geq j\geq 0$, we have, for any $\alpha>0$,
\[
    \Pr_{(\mathsf{RO},X)}\left[\abs{Q_i^{(k)}\cap C}\geq \alpha\right]\leq 2^{-(\alpha(u-\log q-\log v)-s-1)},
\]
where $\mathsf{RO}$ and $X$ are uniformly distributed.
\end{lma}
\begin{proof}
Given a subset $C$, a machine index $i$ and a particular round $k$, consider the massively parallel computation running until the beginning of round $k$ and the memory state $M_i^{(k)}$ given to machine $i$ at the beginning of round $k$. We can use the algorithm $\mathcal{A}_1$ in Lemma~\ref{smallint} to simulate the computation done by $\mathcal{A}$ until the beginning of round $k$ and hence, let $M=M_i^{(k)}$ and $\mathcal{A}_2$ be the computation done by machine $i$ in round $k$. The lemma follows by applying Lemma~\ref{smallint}.
\end{proof}
Lemma~\ref{mrsmallint} only bounds the number of intersection. It is still possible that the algorithm somehow obtain the last answer in only one round. The following lemma helps us rule out this possibility, which says that any algorithm can only query the $j+1$-th entry in the $\SimLine$ with small probability if it has not queried the $j$-th entry. We state it as follows. 

\begin{lma}\label{nojump}
For any deterministic massively parallel computation $\mathcal{A}$, any index $0\leq j\leq w-1$ and any index of query $k$, let $E_{j,k}$ be the event that $\mathcal{A}$ successfully queries $(x_{j+1}, r_{j+1})$ on its $k$-th query, given that all the previous queries $q_i\neq (x_j,r_j)$. Then we have
\[
    \Pr_{(\mathsf{RO},X)}\left[E_{j,k}\right]\leq 2^{-u},
\]
where $\mathsf{RO}$ and $X$ are uniformly distributed.
\end{lma}
\begin{proof}
Suppose all the previous queries are $q_1,...,q_{k-1}$ and denote the next query $q_k$. Suppose further $\mathsf{RO}(x_{i-1},r_{i-1})=r_i$ for $1\leq i\leq j$.
We consider the set of random oracles, $\mathsf{RO}'$, consistent with the answers to the queries $q_1,q_2,...,q_{k-1}$ and the pre-fixed oracle answers $r_1,...,r_j$. For these oracles, the oracle input $(x_{j},r_{j})$ is well-defined and hence, further consider the answer to this input is lazily assigned. Since the answers to the queries are fixed and $\mathcal{A}$ only depends on the answers to its queries, the next query $q_k$ will be the same for any oracle in $\mathsf{RO}'$. Moreover, $r_{j+1}$ is still uniform over all $2^u$ possible values. Hence, we conclude that the guessing probability will be less than $2^{-u}$.
\end{proof}
Now we are able to prove Lemma~\ref{lmasimline}.
\begin{proof}[Proof of Lemma~\ref{lmasimline}]
Let $h=\frac{s}{u-\log q-\log v}+1$. Recall that for each $0\leq j\leq w/h-1$, $C_j=\{(x_{i\Mod v},r_i)\mid jh+1\leq i\leq \min(jh+v,w)\}$.
For each round $k$, let $C^{(k)}=\{(x_{i\Mod v}, r_i)\mid kh+1<i\leq w\}$. For each round $k$, let $Q^{(\leq k)}$ be the set of queries done by all machines until the end of round $k$, $Q^{(k)}$ be the set of queries done by all machines in round $k$, and recall that $Q^{(k)}_i$ is the set of queries done by machine $i$ in round $k$. We show the following claim.
\begin{cla}
For any deterministic massively parallel computation with $m$ machines, local memory of size $s$ and the number of queries $q$ computing $\SimLine_{n,w,u,v}^\mathsf{RO}$ and running until the end of round $k< \frac{w}{h}-1$, 
\[\begin{split}
    & \Pr_{(\mathsf{RO},X)}[\abs{Q^{(\leq k)}\cap C^{(k+1)}}>0] \leq  (k+1)(m2^{-(u-\log q-\log v)}+wmq 2^{-u}).
    \end{split}
\]
\end{cla}
\begin{proof}
Let $E_i^{(k)}$ be the event that, in round $k$, there exist $t\in [q]$ and $j\in [w]$ such that machine $i$ successfully queries $(x_{j+1}, r_{j+1})$ on its $t$-th query given that all the previous queries $q_a\neq (x_j,r_j)$. Then, by Lemma~\ref{nojump} and a union bound, for every $i$ and every $k$,
\[
    \Pr_{(\mathsf{RO},X)}[E_i^{(k)}]\leq wq 2^{-u}.
\]
We prove the claim by induction. For the base case, $k=0$, we know that by Lemma~\ref{smallint} and setting $\alpha=h$, for any machine $i$,
\[
    \Pr_{(\mathsf{RO},X)}[\abs{Q_i^{(0)}\cap C_0}\geq h]\leq 2^{-(u-\log q-\log v)}.
\]
Thus, 
{\small
\begin{align*}
    \Pr_{(\mathsf{RO},X)}[\abs{Q^{(0)}_i\cap C^{(1)}}>0] \leq & \Pr[\abs{Q_i^{(0)}\cap C^{(1)}}>0 \land \overline{E_i^{(0)}}]+\Pr[E_i^{(0)}]\\
    \leq  & \Pr[\abs{Q^{(0)}_i\cap T_0}\geq h \land \overline{E_i^{(0)}}]+\Pr[E_i^{(0)}]\\
    \leq & 2^{-(u-\log q-\log v)}+wq 2^{-u}.
\end{align*}
}
And hence, by a union bound, we have
\[\begin{split}
    & \Pr_{(\mathsf{RO},X)}[\abs{Q^{(\leq 0)}\cap C^{(1)}}> 0]  \leq m2^{-(u-\log q-\log v)}+wmq 2^{-u}.
    \end{split}
\]
Now assume that in round $k-1$, the claim holds. Similarly, by Lemma~\ref{simline} and setting $\alpha=h$, for any machine $i$,
\[
    \Pr_{(\mathsf{RO},X)}[\abs{Q_i^{(k)}\cap C_k}\geq h]\leq 2^{-(u-\log q-\log v)}.
\]
 Thus,
 \ifnum\submission=0
 {\small
\begin{align*}
    \Pr_{(\mathsf{RO},X)}[\abs{Q^{(\leq k)}\cap C^{(k+1)}}>0]
    &\leq \Pr[\abs{Q^{(\leq k-1)}\cap C^{(k)}}>0]+\Pr[\abs{Q^{(\leq k)}\cap C^{(k+1)}}>0\land \abs{Q^{(\leq k-1)}\cap C^{(k)}}=0]\\
    &= \Pr[\abs{Q^{(\leq k-1)}\cap C^{(k)}}>0]+\Pr[\abs{Q^{(k)}\cap C^{(k+1)}}>0]\\
    &\leq \Pr[\abs{Q^{(\leq k-1)}\cap C^{(k)}}>0]+m\Pr[\abs{Q_i^{(k)}\cap C^{(k+1)}}>0]\\
    &\leq \Pr[\abs{Q^{(\leq k-1)}\cap C^{(k)}}>0]+m\Pr[\abs{Q_i^{(k)}\cap C_k}\geq h]+m\Pr[E_i^{(k)}]\\
    &\leq k(m2^{-(u-\log q-\log v)}+wmq 2^{-u})+m2^{-(u-\log q-\log v)}+wmq 2^{-u}\\
    &=(k+1)(m2^{-(u-\log q-\log v)}+wmq 2^{-u}).
\end{align*}
}
\else
{\small
\begin{align*}
    &\Pr_{(\mathsf{RO},X)}[\abs{Q^{(\leq k)}\cap C^{(k+1)}}>0]\\
    \leq & \Pr[\abs{Q^{(\leq k-1)}\cap C^{(k)}}>0]\\
    &+\Pr[\abs{Q^{(\leq k)}\cap C^{(k+1)}}>0\land \abs{Q^{(\leq k-1)}\cap C^{(k)}}=0]\\
    = & \Pr[\abs{Q^{(\leq k-1)}\cap C^{(k)}}>0]+\Pr[\abs{Q^{(k)}\cap C^{(k+1)}}>0]\\
    \leq & \Pr[\abs{Q^{(\leq k-1)}\cap C^{(k)}}>0]+m\Pr[\abs{Q_i^{(k)}\cap C^{(k+1)}}>0]\\
    \leq & \Pr[\abs{Q^{(\leq k-1)}\cap C^{(k)}}>0]+m\Pr[\abs{Q_i^{(k)}\cap T_k}\geq h]\\
    &+m\Pr[E_i^{(k)}]\\
    \leq & k(m2^{-(u-\log q-\log v)}+wmq 2^{-u})+m2^{-(u-\log q-\log v)}\\
    &+wmq 2^{-u}\\
    = &(k+1)(m2^{-(u-\log q-\log v)}+wmq 2^{-u}).
\end{align*}
}
\fi
\end{proof}
    Let $\textbf{Success}$ be the event that $\alpha$ successfully compute $\SimLine_{n,w,u,v}^\mathsf{RO}$. To compute $\SimLine_{n,w,u,v}^\mathsf{RO}$, the algorithm must reach $(x_w,r_w)$. However, $(x_w,r_w)\in C^{(\frac{w}{h}-1)}$ and hence, by our claim,
\begin{align*}
    \Pr_{(\mathsf{RO},X)}[\textbf{Success}]
    &\leq \Pr_{(\mathsf{RO},X)}[|Q^{(\leq \frac{w}{h}-2)}\cap C^{(\frac{w}{h}-1)}|>0]\\
    &\leq \frac{w}{h}(m2^{-(u-\log q-\log v)}+wmq 2^{-u})\\
    &\leq 2^{-\Omega(u-\log q-\log v-\log m-\log w)}.
\end{align*}
As $n$ becomes sufficiently large, the success probability becomes sufficiently small.
\end{proof}

\ifnum\submission=1
\end{multicols}
\fi
\end{document}